\documentclass[11pt]{article}
\usepackage{verbatim}
\usepackage{xcolor}
\usepackage[utf8]{inputenc}
\usepackage{fancyhdr}
\usepackage[letterpaper, portrait, margin=1in]{geometry}
\usepackage[document]{}
\usepackage{graphicx}
\usepackage{enumitem}
\usepackage{subcaption}
\usepackage{caption}
\usepackage{multirow}
\usepackage{float}
\usepackage{relsize}
\usepackage{amsmath}
\usepackage{amssymb}
\usepackage{amsthm}
\usepackage{mathtools}
\usepackage{ellipsis}
\usepackage{bbm}
\usepackage[T1]{fontenc}
\usepackage{threeparttable}
\usepackage{fontspec}
\usepackage{microtype}
\usepackage{booktabs}
\usepackage{mathrsfs}
\usepackage{afterpage}
\usepackage{listings}
\usepackage{tikz}
\usetikzlibrary{arrows.meta,positioning,decorations.pathreplacing}
\usepackage[ruled,vlined]{algorithm2e}
\usepackage[sortcites=true, backref=true, backend=biber, style=alphabetic, maxcitenames=999, maxbibnames=999, maxalphanames=999, sorting=nyt]{biblatex}
\usepackage[colorlinks=true, linkcolor=darkgray, citecolor=gray, urlcolor=black]{hyperref}
\usepackage[capitalise,noabbrev,nameinlink]{cleveref}

\usepackage[most]{tcolorbox}

\newtheorem{theorem}{Theorem}[section]
\newtheorem{lemma}[theorem]{Lemma}

\newtheorem{invariant}[theorem]{Invariant}

\tcolorboxenvironment{theorem}{%
  enhanced,
  breakable,
  colback=white,
  colframe=black!70,
  coltitle=black,
  fonttitle=\bfseries,
  boxrule=0.7pt,
  arc=2pt,
  boxsep=4pt,
  left=6pt,right=6pt,top=4pt,bottom=4pt,
}

\tcolorboxenvironment{lemma}{%
  enhanced,breakable,
  colback=white,
  colframe=black!70,
}
\tcolorboxenvironment{corollary}{%
  enhanced,breakable,
  colback=white,
  colframe=black!70,
}
\tcolorboxenvironment{proposition}{%
  enhanced,breakable,
  colback=white,
  colframe=black!70,
}
\tcolorboxenvironment{invariant}{%
  enhanced,breakable,
  colback=white,
  colframe=black!70,
}
\tcolorboxenvironment{definition}{%
  enhanced,breakable,
  colback=blue!2,
  colframe=blue!60!black,
}
\tcolorboxenvironment{remark}{%
  enhanced,breakable,
  colback=gray!5,
  colframe=gray!70,
  fonttitle=\itshape,
}

\AtBeginRefsection{\GenRefcontextData{sorting=ynt}}
\AtEveryCite{\localrefcontext[sorting=ynt]}
\DeclareMathOperator*{\argmax}{arg\,max}
\DeclareMathOperator*{\argmin}{arg\,min}

\newcommand{\dpw}{(\Delta + 1)}

\title{Dynamic $\dpw$ Vertex Coloring}
\author{Noam Benson-Tilsen}
\date{12 January 2026}

\begin{document}

\setmainfont{texgyrepagella-regular.otf}
[
    BoldFont = texgyrepagella-bold.otf ,
    ItalicFont = texgyrepagella-italic.otf ,
    BoldItalicFont = texgyrepagella-bolditalic.otf
]

\maketitle

\newcommand{\C}{\mathbb{C}}
\newcommand{\R}{\mathbb{R}}
\newcommand{\Q}{\mathbb{Q}}
\newcommand{\Z}{\mathbb{Z}}
\newcommand{\N}{\mathbb{N}}
\newcommand{\E}{\mathbb{E}}
\newcommand{\Ep}{\mathcal{E}}
\newcommand{\st}{~|~}
\newcommand{\given}{~|~}
\newcommand{\divides}{\bigm|}
\newcommand{\seq}[2]{\{ #1_{#2} \}}
\newcommand{\compl}[1]{\mathcal{C}#1}
\newcommand{\1}{\mathbbm{1}}

\newcommand{\ceil}[1]{\lceil #1 \rceil}

\newcommand{\fa}{~\forall}
\newcommand{\inv}[1]{#1^{-1}}
\newcommand{\kernel}[1]{\text{ker}(#1)}
\newcommand{\Aut}{\text{Aut}}
\newcommand{\Inn}{\text{Inn}}
\newcommand{\degree}[1]{\text{deg}(#1)}

\newcommand{\todo}[1]{\textcolor{red}{#1}}

\newcommand{\A}{\mathcal{A}}
\newcommand{\Adv}{\forall}
\newcommand{\Lo}{\mathcal{L}}
\newcommand{\Hi}{\mathcal{H}}
\newcommand{\Sa}{\mathcal{S}}
\newcommand{\B}{\mathcal{B}}
\newcommand{\Coup}{\mathcal{C}}
\newcommand{\U}{\mathcal{U}}
\newcommand{\M}{\mathcal{M}}
\newcommand{\D}{\mathcal{D}}

\DeclarePairedDelimiter\abs{\lvert}{\rvert}%
\DeclarePairedDelimiter\sizeof{\lvert}{\rvert}%
\DeclarePairedDelimiter\norm{\lVert}{\rVert}%

\makeatletter
\let\oldsizeof\sizeof
\def\sizeof{\@ifstar{\oldsizeof}{\oldsizeof*}}
\let\oldabs\abs
\def\abs{\@ifstar{\oldabs}{\oldabs*}}
\let\oldnorm\norm
\def\norm{\@ifstar{\oldnorm}{\oldnorm*}}
\makeatother
%

\begin{abstract}
    \noindent
    Several recent results from dynamic and sublinear graph coloring are surveyed.
    This problem is widely studied and has motivating applications like network topology control, constraint satisfaction, and real-time resource scheduling.
    Graph coloring algorithms are called \textit{colorers}.
    In \S \ref{sec:prelim} are defined graph coloring, the dynamic model, and the notion of performance of graph algorithms in the dynamic model.
    In particular $\dpw$-coloring, sublinear performance, and oblivious and adaptive adversaries are noted and motivated.
    In \S \ref{sec:warmup} the pair of approximately optimal dynamic vertex colorers given in \cite{barba2017dynamic} are summarized as a warmup for the $\dpw$-colorers.
    In \S \ref{sec:dyn} the state of the art in dynamic $\dpw$-coloring is presented. 
    This section comprises a pair of papers (\cite{bhattacharya2018dynamic} and \cite{bhattacharya2022fully}) that improve dynamic $\dpw$-coloring from the naive algorithm with $O(\Delta)$ expected amortized update time to $O(\log \Delta)$, then to $O(1)$ with high probability. 
    In \S \ref{sec:adaptive} the results in \cite{behnezhad2025fully}, which gives a sublinear algorithm for $\dpw$-coloring that generalizes oblivious adversaries to adaptive adversaries, are presented.
\end{abstract}

\setcounter{tocdepth}{3}
\tableofcontents
    
\section{Preliminaries}
\label{sec:prelim}
    \subsection{Graph coloring}
        A vertex coloring $\chi : V \longrightarrow C$ of a graph $G \coloneqq (V, E)$ is an assignment of the vertices $V$ to colors $\Coup \coloneqq \{ 1, \ldots, k \} \subseteq \N_+$. 
        $\chi$ is proper if $\chi(u) \neq \chi(v)$ for all edges $\{ u, v \} \in E$. 
        $G$ is $k$-colorable if $\chi$ is a proper coloring of $G$ with largest color $\max_{c \in C} c = k$. The smallest $k$ such that there exists a proper coloring of $G$ is the chromatic number $\chi(G)$.

        In general, however, determining $\chi(G)$ is NP-hard, even for $\chi(G) = 3$.
        This motivates studying approximation algorithms (\S \ref{sec:warmup}), as well as algorithms for finding $\dpw$-colorings (\S \ref{sec:dyn} and on), which always exist.
        See \S \ref{subsec:perf} for more basic facts about $\dpw$-coloring performance.
    
    \subsection{Dynamic model}
        In the dynamic model, $G$ undergoes a sequence of updates $\U$ consisting of insertions and deletions of edges and vertices. 
        Note that neither vertex deletions nor edge deletions corrupt a proper coloring.
        In the context of vertex coloring, one can choose between accounting solely for (1) vertex (plus incident edge) insertions or (2) edge insertions.
        \begin{enumerate}
            \item 
                Edge updates can be factored out by
                \begin{enumerate}
                    \item 
                        simulating any edge insertion by removing and adding back one of its endpoints, and
                    \item 
                        delaying accounting for any edge deletion until one of its endpoints is deleted; or
                \end{enumerate}
            \item 
                vertex updates can be factored out by assuming (WLOG) that $G$ begins with the maximal order $n$ it would have at any point during $\U$, in which case vertex updates are the same as edge updates with respect to maintaining a proper vertex coloring.
        \end{enumerate} 
        In \S \ref{sec:warmup} the first approach is taken; in \S \ref{sec:dyn}, the second.

        Maintaining a proper coloring across $\U$ requires recoloring vertices; i.e., upon an edge insertion a dynamic graph colorer must update the current coloring $\chi$ by modifying the colors $\chi(V^\prime)$ of some vertices $V^\prime \subseteq V$. The maximum degree $\stackrel{\U}{\max_{v \in V(G^{\U})}} d(V)$ across $\U$ of any vertex is denoted by $\Delta$.

    \subsection{Performance}
    \label{subsec:perf}
        The performance of a dynamic graph colorer is defined by its (expected amortized) number of recolorings per update, also known as update time, as follows.
        For some underlying set $\mathscr{U}$ of inputs, let update time $T_{\A}(\U)$ be the performance of dynamic colorer $\A$ against an input sequence $\U \sim \mathfrak{S}_{\mathscr{U}}$ drawn from the symmetric group $\mathfrak{S}_{\mathscr{U}}$ of $\mathscr{U}$.\footnote{When considering only edge insertions, as in most of the report, $T$ will specialize to $T(\U, n)$, where $n$ is the order of the initially empty graph.}
        Then the task is to find 
        \begin{align*}
            \argmin_{\A} \mathop{\E}\limits_{\mathbf{U} \sim \mathfrak{S}_{\mathscr{U}}} [T_{\A}(\U)].
        \end{align*}
        Linear-time dynamic $\dpw$-colorers are trivial, but sublinear dynamic $\dpw$-colorers are nontrivial.
        This is the case even when inputs are oblivious --- that is, when input sequences are generated independently of the colorer, such as being drawn uniformly at random from $\mathfrak{S}_{\mathscr{U}}$.
        The end of this report will discuss a paper that gives sublinear performance against adaptively adversarial input sequences, which can force pathological, worst-case behavior based on the colorer.
        In \cite{barba2017dynamic}, presented in \S \ref{sec:warmup}, the colorers use $k(N) \cdot \chi(G)$ colors for a parameter $k > 0$ depending on the maximum number $N$ of vertices of graph $G$ across all updates.
        The expected amortized number of recolorings per update is traded against approximation tightness in the two colorers.
        In particular, for a $k$-colorable graph with parameter $d > 0$, two algorithms are given in which a factor in $O(N^{1/d})$ is traded between performance and approximation tightness.
        
\section{Warmup: Approximately optimal dynamic coloring}
\label{sec:warmup}

    \cite{barba2017dynamic} provides two variations of the same algorithm. 
    The first algorithm, $\A_1$, uses a one-level vertex partition; $\A_2$ uses a two-level partition.\footnote{The paper presents $\A_1$ and $\A_2$ in the other order, which is less clear. I also combine their presentation because they have one difference.} 
    In both, vertices within a first-level partition are colored properly using $\chi(G)$ colors.
    Also in both, $V$ is split into buckets $\bigoplus_{i = 1}^d V_i$ such that $V_i$ has capacity $N_R^{(i - 1) / d}$, where $N_R$ is the value of $n$ at the last update. 
    Note that 
    \begin{align*}
        \sum_{i = 1}^d N_R^{(i - 1) / d} = \frac{N_R - 1}{N_R^{1/d} - 1} \in O(N_R^{1 - 1/d}) \subseteq O(n).
    \end{align*}
    Buckets, however, do not reach their capacity.
    Instead, $V_i$ has an invariant high point of at most $h_i = N_R^{i / d} - N_R^{(i - 1) / d}$ vertices.
    In $\A_2$, each first-level bucket $V_i$ has its up to $h_i$ vertices partitioned into $N_R^{1 / d} - 1$ buckets $V_{i, 1}, \ldots, V_{i, \ceil{N_R^{1 / d} - 1}}$ each of capacity $N_R^{(i - 1) / d}$, with an extra empty \emph{reset bucket}. 
    That is, in $\A_2$, sub-buckets are left-packed to leave room for the reset bucket. 
    The reset bucket is used during insertions. 

    \subsection{Vertex insertion}
        When vertex $v$ is inserted, it is placed in $V_1$. In $\A_2$, $v$ goes into the first empty sub-bucket $V_{1, j}$. If $V_1$ has fewer than $h_i$ vertices, $v$ is assigned one of the leftover colors local to $V_1$. In $\A_2$, this color is local to $V_{1, j}$. 
        But if inserting $v$ to $V_1$ violates the high-point invariant, all vertices in $V_1$ are moved to $V_2$ (or in $\A_2$, moved to the empty sub-bucket, guaranteed to exist by the invariant). 
        The vertices that were transferred are colored again, if possible; if they violate the invariant in $V_2$ as well, then $V_2$ is transferred up to $V_3$. 
        This can propagate up to the last partition, $V_d$. 
        In this case, the entire graph is reset and recolored.

    \subsection{Performance}
        In the worst case, $\A_1$ requires shifting and recoloring $O(N^{1 / d})$ vertices per bucket across all $d$ buckets. 
        Each recoloring is trivial, given that there are more than enough colors to assign each vertex a distinct color. 
        In total gives recolorings in $O(dN^{1 / d})$. Similarly, in $\A_2$, vertices are moved and trivially recolored at most once per first-level partition, giving total recolorings in $O(d)$.

        For $\A_1$, each bucket uses $O(\chi(G))$ colors, and there are $O(d)$ buckets, giving an $O(d)$-approximate coloring. 
        For $\A_2$, each bucket within a partition uses $O(\chi(G))$ colors, and with $O(N^{1 / d})$ buckets in each of the $d$ partitions, this gives an $O(dN^{1 / d})$-approximate coloring. 

        \begin{table}[h]
            \centering
            \begin{threeparttable}
                \caption{Performance and tightness}
                \label{tab:performance}
                \begin{tabular}{@{}lccc@{}}
                    \toprule
                    & \textbf{Recolorings} & \textbf{Approximation tightness} \\ 
                    \midrule
                    $\mathbf{\A_1}$ & $O(dN^{1 / d})$ & $O(d)$ \\
                    $\mathbf{\A_2}$ & $O(d)$ & $O(dN^{1 / d})$ \\
                    \bottomrule
                \end{tabular}
            \end{threeparttable}
        \end{table}

\section{Dynamic $\dpw$-coloring with oblivious adversaries}
\label{sec:dyn}
    In this section, I'll give an overview of the techniques used in \cite{bhattacharya2018dynamic} and \cite{bhattacharya2022fully}. For the remainder all updates will be edge insertions, and $k \coloneqq \Delta + 1$, so that $\Coup = \{ 1, \ldots, \Delta + 1 \}$. $\chi^*$ will denote the coloring during $\U$; $\chi^* = \chi$ exactly when $\U$ finishes, and between atomic steps of $\A$ $\chi^*$ is not guaranteed to be proper.

    \subsection{Logarithmic-time dynamic $\dpw$-coloring}
    \label{subsec:log_dyn_d+1}
        There is a trivial algorithm with $O(\Delta)$ update time. 
        When edge $uv$ is inserted, if $\chi^*(u) = \chi^*(v)$, choose one $x$ of $u$ and $v$ and assign $x$ a color not used by any of its neighbors. 
        Such a color is a blank color of $x$, and the set of such colors is $\B_{N(x)}$.
        Note that for any $x \in V$, $\sizeof{N(x)} \leq \Delta < \Delta + 1 = \sizeof{C}$, so by the pigeonhole principle a blank color always exists.

        \subsubsection{Level data structure}
        
        To improve update times to $O(\log \Delta)$, \cite{bhattacharya2018dynamic} presents $\A_{\log}$, a dynamic colorer with logarithmic update time that uses a level data structure $\ell : V \longrightarrow \{ 4, \ldots, L \}$, where $L \coloneqq \log_\beta \Delta$ for sufficiently large constant $\beta > 0 \in \N$. 
        In other words, $\ell(v)$ is the level of $v$, where there are $O(\log \Delta)$ possible levels. 
        The motive for $\ell$ is to expose a heuristic for which colors to use during recoloring that minimizes the expected length of cascading recoloring chains. 
        For each vertex $v$, let the down-neighbors $\Lo(v)$ of $v$ be the neighbors of $v$ with levels lower than $\ell(v)$, the up-neighbors $\Hi(v)$ those with level at least $\ell(v)$, and same-level neighbors $\Sa(v)$ those with level equal to $\ell(v)$. 
        See Figure \ref{fig:LDSlog}.
        $\A_{\log}$ maintains two invariants:
        \begin{invariant}
        \label{inv:down}
            The number $\phi \coloneqq \sizeof{\Lo(v)}$ of down-neighbors is at least $\beta^{\ell(v) - 5} \in \Omega(\Delta)$ for each $v \in V$.
        \end{invariant}
        \begin{invariant}
        \label{inv:up}
            The number $\sizeof{\Lo(v) \cup \Sa(v)}$ of down- and same-level neighbors is at most $\beta^{\ell(v)} \in O(\Delta)$ for each $v \in V$.
        \end{invariant}

        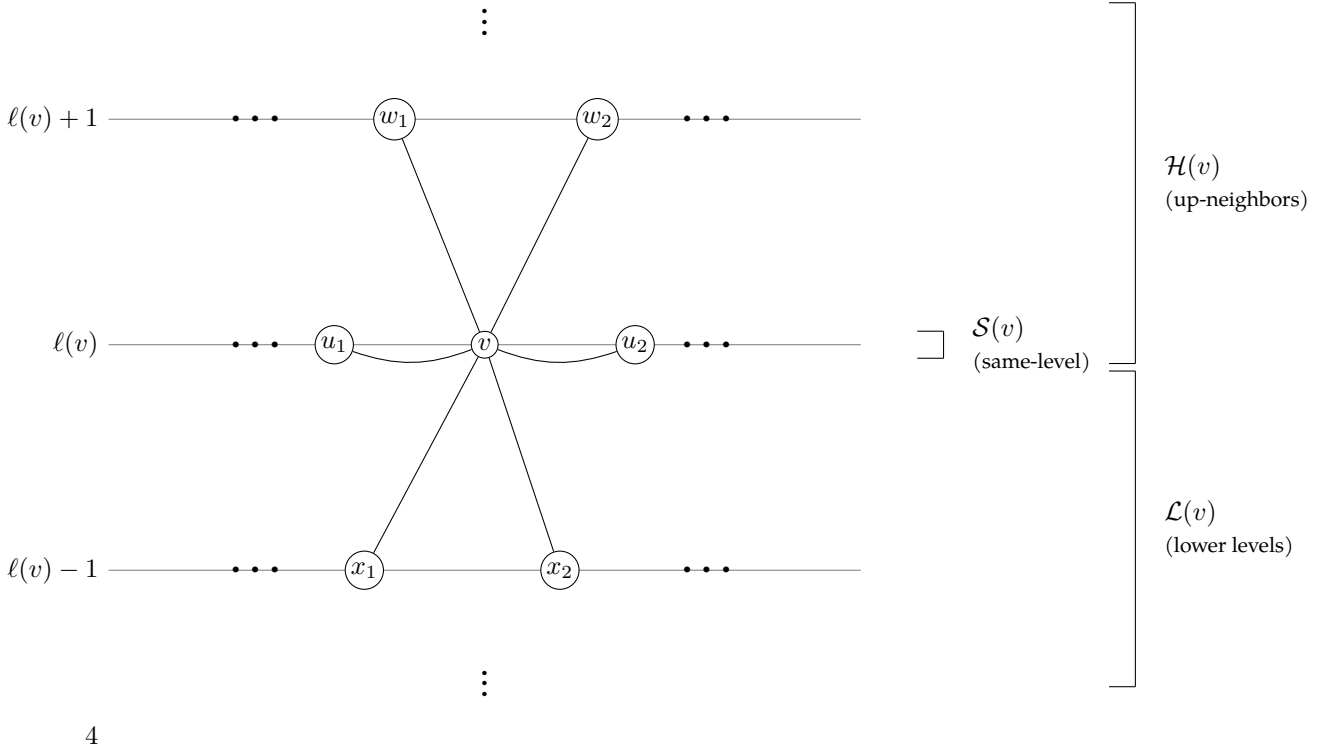
\begin{figure}[t]
            \centering
            \begin{tikzpicture}[x=1cm,y=1cm,>=Latex, font=\small]
                
                \def\X{-5}
                \def\Xr{5}
                \def\yTop{3}
                \def\yMid{0}
                \def\yBot{-3}
                
                \def\Bx{6.2}        
                \def\Bin{5.75}      
                \def\tick{0.35}     
                
                \newcommand{\RBracket}[3]{
                \draw (#1,#2) -- (#1,#3);
                \draw (#1,#2) -- ++(-\tick,0);
                \draw (#1,#3) -- ++(-\tick,0);
                }
                
                \node[anchor=east] at (\X,\yTop+2.2) {$L = \log_\beta \Delta$};
                \node at (0,\yTop+1.4) {\Large $\boldsymbol{\vdots}$};
                
                \foreach \y/\lab in {
                \yTop/{\ell(v)+1},
                \yMid/{\ell(v)},
                \yBot/{\ell(v)-1}
                }{
                \draw[gray] (\X,\y) -- (\Xr,\y);
                \node[anchor=east] at (\X,\y) {$\lab$};
                
                \node at (-3,\y) {\Large $\boldsymbol{\cdots}$};
                \node at ( 3,\y) {\Large $\boldsymbol{\cdots}$};
                }
                
                \node at (0,\yBot-1.4) {\Large $\boldsymbol{\vdots}$};
                \node[anchor=east] at (\X,\yBot-2.2) {$4$};
                
                \node[circle,draw,fill=white,inner sep=1.2pt] (v) at (0,\yMid) {$v$};
                
                \node[circle,draw,fill=white,inner sep=1.0pt] (h1) at (-2,\yMid) {$u_1$};
                \node[circle,draw,fill=white,inner sep=1.0pt] (h2) at ( 2,\yMid) {$u_2$};
                
                \node[circle,draw,fill=white,inner sep=1.0pt] (up1) at (-1.2,\yTop) {$w_1$};
                \node[circle,draw,fill=white,inner sep=1.0pt] (up2) at ( 1.5,\yTop) {$w_2$};
                
                \node[circle,draw,fill=white,inner sep=1.0pt] (dn1) at (-1.6,\yBot) {$x_1$};
                \node[circle,draw,fill=white,inner sep=1.0pt] (dn2) at ( 1.0,\yBot) {$x_2$};
                
                \draw (v) -- (up1);
                \draw (v) -- (up2);
                \draw (v) -- (dn1);
                \draw (v) -- (dn2);
                
                \draw[bend left=20]  (v) to (h1);
                \draw[bend right=20] (v) to (h2);
                
                \def\LabelX{6.45} 
                
                \def\Bx{8.65}          
                \def\LabelX{\Bx+0.25}  
                
                \def\pad{2.55}         
                \def\labelpad{2.55}    
                \def\Bin{\Bx-\pad}     
                \def\LabelXin{\LabelX-\labelpad} 

                \RBracket{\Bx}{\yMid-0.25}{\yTop+1.55}
                \node[anchor=west,align=justify] at (\LabelX, {(\yMid-0.25+\yTop+1.55)/2})
                {$\Hi(v)$\\{\scriptsize(up-neighbors)}};
                
                \RBracket{\Bin}{\yMid-0.18}{\yMid+0.18}
                \node[anchor=west,align=justify] at (\LabelXin, \yMid)
                {$\Sa(v)$\\{\scriptsize(same-level)}};

                \RBracket{\Bx}{\yBot-1.55}{\yMid-0.35}
                \node[anchor=west,align=justify] at (\LabelX, {(\yBot-1.55+\yMid-0.35)/2})
                {$\Lo(v)$\\{\scriptsize(lower levels)}};
            
            \end{tikzpicture}
            \caption{$\ell(v)$ for each $v \in V$, with up-neighbors $\Hi(v)$, same-level neighbors $\Sa(v) \subseteq \Hi(v)$, and down-neighbors $\Lo(v)$.}
            \label{fig:LDSlog}
        \end{figure}

        Let $\Coup_{\Hi(v)}$ be the colors used by an up-neighbor of $v$, $\U_{\Lo(v)}$ those used by exactly one down-neighbor of $v$, and $\M_{\Lo(v)}$ those used by at least two down-neighbors of $v$. 
        Let $\D_{N(v)} \coloneqq \Coup \setminus \Coup_{\Hi(v)} \setminus \M_{\Lo(v)}$ be the colors used by no up-neighbor of $v$ and at most one down-neighbor of $v$. Let $\tau_v$ be the last timestamp at which $v$ was recolored.
        See Table \ref{tab:objects} for these definitions. Sets are denoted by calligraphic characters. Sets of vertices are named by functions $V \longrightarrow \mathbb{P}(V)$ and sets of colors have names subscripted by the set of vertices over which they are defined.

        \afterpage{
            \begin{table}[t]
                \centering
                \begin{threeparttable}
                    \caption{Data for level data structure for each $v \in V$}
                    \label{tab:objects}
                    \begin{tabular}{@{}lccc@{}}
                        \toprule
                        & \textbf{Name} & \textbf{Definition} & \textbf{Description} \\ 
                        \midrule
                        & $L$ & $\log_\beta \Delta$ & Highest level \\
                        & $\ell(v)$ & $\in \{ 4, \ldots, L \}$ & Level of $v$ \\
                        & $\deg(v)$ & $\sizeof{N(v)}$ & Degree of $v$ \\
                        & $\tau_v$ & $\in \N$ & Last time at which $v$ was recolored \\
                        & $\Lo(v)$ & $\{ u \st \{ u, v \} \in E \wedge \ell(u) < \ell(v) \}$ & Neighbors of $v$ of lower level \\
                        & $\Hi(v)$ & $\{ u \st \{ u, v \} \in E \wedge \ell(u) \geq \ell(v) \}$ & Neighbors of $v$ with at least $v$'s level \\
                        & $\Sa(v)$ & $\{ u \st \{ u, v \} \in E \wedge \ell(u) = \ell(v) \}$\tnote{1} & Neighbors of $v$ with the same level as $v$ \\
                        & $\B_{N(v)}$ & $\{ c \st \forall u \in N(v) : \chi^*(u) \neq c \}$ & Colors of no neighbor of $v$ \\
                        & $\Coup_{\Hi(v)}$ & $\{ c \st \exists u \in \Hi(v) \st \chi^*(u) = c \}$ & Colors of an up-neighbor of $v$ \\
                        & $\U_{\Lo(v)}$ & $\{ c \st \exists! ~ u \in \Lo(v) \st \chi^*(u) = c \}$ & Colors of exactly one down-neighbor of $v$ \\
                        & $\M_{\Lo(v)}$ & $\{ c \st \exists_{\geq 2} u \in \Lo(v) \st \chi^*(u) = c \}$\tnote{2} & Colors of at least two down-neighbors of $v$ \\
                        & $\D_{N(v)}$ & $\Coup \setminus \Coup_{\Hi(v)} \setminus \M_{\Lo(v)}$ & Colors of no up-neighbor and at most one down-neighbor of $v$ \\
                        \bottomrule
                    \end{tabular}
                    \begin{tablenotes}
                        \item[1]{Note that this is a subset of $\Hi(v)$.}
                        \item[2]{Note that this is $\Coup \setminus \B_{N(v)} \setminus \Coup_{\Hi(v)} \setminus \U_{\Lo(v)}$.}
                    \end{tablenotes}
                \end{threeparttable}
            \end{table}
        }
        
        $\D_{N(v)}$ is the most important object. 
        In the worst-case edge insertion, where there might be a recursive cascade of recolorings, $\D_{N(v)}$ will help limit the recursion fan-out to one by exposing colors used by few of the recolored vertex's down-neighbors. 
        By Lemma \ref{lemma:subpalette}, $\D_{N(v)}$ is large enough that sampling uniformly at random from it is unlikely to conflict with the vertex being recolored.
        
        \begin{lemma}
        \label{lemma:subpalette}
            If Invariant \ref{inv:down} and Invariant \ref{inv:up} hold, $\sizeof{\D_{N(v)}} \in \Omega(\Delta)$ for each $v \in V$.
        \end{lemma}
        \begin{proof}
            Let $x \coloneqq \sizeof{\{ v \st \chi^*(v) \in \U_{\Lo(v)} \}}$ be the number of down-neighbors with a color used exactly once in $N(v)$. Note that $\sizeof{\M_{\Lo(v)}} \leq \frac{\phi - x}{2}$: Each of the $\phi$ down-neighbors of $v$, except for the $x$ down-neighbors with a color used exactly once, has at least one partner with the same color. Therefore 
            \begin{align*}
                \sizeof{\D_{N(v)}} \geq \sizeof{\Coup \setminus \Coup_{\Hi(v)}} - \frac{\phi - x}{2} = (\Delta + 1) - (\Delta - \phi) - \frac{\phi - x}{2} \leq \frac{\phi}{2} + 1,
            \end{align*}
            so $\sizeof{\D_{N(v)}} \in \Omega({\phi}) = \Omega({\Delta})$.
        \end{proof}

        \subsubsection{Algorithm}

            \begin{algorithm}[H]
            \caption{\textsc{Edge insertion}\textsubscript{$\log$}}
            \label{alg:insert_log}
                \KwIn{$G$; $\{ u, v \} \leftrightsquigarrow E(G)$ for deletion or insertion; $\ell$ and associated data structures}
                \KwOut{Updated $\dpw$-coloring $\chi*$}
                \BlankLine
                \textsc{Update levels}\;
                \If{$\{ u, v \}$ is being inserted and $\chi^*(u) = \chi^*(v)$}{
                    $x \gets \displaystyle \argmax_{x \in \{u, v\}} \tau_w$\tcp*{More-recently recolored}
                    \textsc{recolor}\textsubscript{$\log$}($x$)\;
                }
            \end{algorithm}

            \begin{algorithm}[H]
            \caption{\textsc{Update levels}}
            \label{alg:UL}
                \KwIn{$G$; $\chi^*$; $\ell$ and associated data structures}
                \KwOut{$\ell$ and associated data structures with Invariant \ref{inv:down} and Invariant \ref{inv:up} satisfied}
                \BlankLine
                \While{Invariant \ref{inv:down} or Invariant \ref{inv:up} is violated}{
                    \If{there exists $x \in V$ that violates Invariant \ref{inv:up} (having more than $\beta^{\ell(x)}$ down-neighbors)}{
                        find the minimum level $k > \ell(x)$ for $x$ that would have $\sizeof{\Lo(x)} + \sizeof{\Sa(x)} \leq \beta^k$\;
                        $\ell(x) \gets k$\;
                        update any auxiliary data structures\;
                    }
                    \Else{
                        find a vertex $x \in V$ that violates Invariant \ref{inv:down} (having fewer than $\beta^{\ell(x) - 5}$ same-level neighbors)\;
                        \If{there exists a level $k < \ell(x)$ for $x$ that would have $\sizeof{\Lo(x)} \geq \beta^{k-1}$}{
                            let $k^\prime$ be the maximum such level\;
                            $\ell(x) \gets k^\prime$\;
                        }
                        \Else{
                           $\ell(x) \gets 4$\;
                        }
                    }
                    update any auxiliary data structures\;
                }
            \end{algorithm}
            
            \begin{algorithm}[H]
            \caption{\textsc{recolor}\textsubscript{$\log$}}
            \label{alg:recolor_log}
                \KwIn{$G$; $v \in V(G)$ for recoloring; $\chi^*$; $\ell$ and associated data structures}
                \KwOut{$\chi^*$ with new color for $v$}
                \BlankLine
                $c \sim \D_{N(v)}$ uniformly at random\;
                $\chi^*(v) \gets c$\;
                Update relevant auxiliary data structures\; 
                \If{$\exists\, w \in \Lo(v) : \chi^*(w) = c$}{
                    \textsc{recolor}\textsubscript{$\log$}($w$)\;
                }
            \end{algorithm}

    \subsubsection{Analysis}

        Insertion (Algorithm \ref{alg:insert_log}) relies on two subroutines: \textsc{Update levels} (Algorithm \ref{alg:UL}) and \textsc{recolor}\textsubscript{$\log$} (Algorithm \ref{alg:recolor_log}). To insert an edge $\{ u, v \}$ to $E(G)$, first the level data structure $\ell$ is repaired so that Invariants \ref{inv:down} and \ref{inv:up} are satisfied, then $\{ u, v \}$ is added to $E(G)$. If $\chi^*(u) = \chi^*(v)$, the one more recently recolored is recolored. An auxiliary list consisting of invariant-violating \textit{dirty} vertices is also maintained.

        The idea in \textsc{Update levels} is to move vertices in $\ell$ to ensure the invariants hold before repairing $\chi^*$ if needed. First, as long as there is a vertex $x$ with too many (more than $\beta^{\ell(x)}$) down-neighbors, $x$ is moved up to the lowest level such that Invariant \ref{inv:up} holds of it at that level. Note that $x$ moves up levels in $\ell$, the number of down-neighbors of $x$ increases more slowly than $\beta^{\ell(x)}$; in particular it is always possible to fix $x$ this way. On the other hand, promoting $x$ might make some other vertices dirty, but by the end of the while loop all vertices are clean with respect to Invariant \ref{inv:up}. (For a proof of this, see \cite{bhattacharya2018dynamic} Lemma 3.1.) After cleaning each vertex with respect to Invariant \ref{inv:down}, \textsc{Update levels} cleans the vertices with respect to Invariant \ref{inv:down}. Namely, while there is a dirty vertex $x$, $x$ is moved down to the highest level such that it is clean. The properties of the cleaning loop for Invariant \ref{inv:up} also hold in this loop. 
        
        \textsc{Update levels} is the source of \textsc{Edge insertion}\textsubscript{$\log$}'s $O(\log \Delta)$ runtime. Details of the proof of \textsc{Update level}'s $O(\beta) = O(\log \Delta)$ running time are in \cite{bhattacharya2022fully} Theorem 3.2. More importantly, the other component of \textsc{Edge insertion}\textsubscript{$\log$} --- \textsc{recolor}\textsubscript{$\log$} --- is implementable in constant expected amortized time.

        In \textsc{recolor}\textsubscript{$\log$}, if $\chi^*(u) \neq \chi^*(v)$ then $\chi^*$ remains proper. Otherwise, the more-recently recolored vertex $x$ of $u$ and $v$, tracked by $\tau_u, \tau_v$, is recolored. To recolor $x$, $\chi^*(x)$ is set to a random sample $c \sim \D_{N(x)}$. Note that each color in $\D_{N(x)}$ is assigned by $\chi^*$ to at most one neighbor $w$ of $x$. If $c = \chi^*(w)$, then $\{ x, w \}$ is a monochromatic edge. To fix this corruption, \textsc{recolor}\textsubscript{$\log$} recurses on $w$. Observe that this can cascade down $O(\beta^{\ell(x)})$ levels, each requiring a sample from $\D_{N(x)}$ and a scan of $\Lo(x)$ for $w$. Sampling is difficult technically because $\sizeof{\D_{N(x)}} \gg \beta^{\ell(x)}$ in general, but truncating the sample at the first $\beta^{\ell(x)}$ elements of $\D_{N(x)}$ works.

        In any case, each initial call to \textsc{recolor}\textsubscript{$\log$} takes time in $O(\beta^{\ell(x)}) = O(\Delta)$. To show that \textsc{recolor}\textsubscript{$\log$} has expected amortized time in $O(1)$, it is necessary to show that it is called with sufficiently low probability given $\chi^*$ and $x$'s new color. By Lemma \ref{lemma:subpalette} and the definition of $\D_{N(x)}$, at most one of $\D_{N(x)}$'s $\Omega(\Delta)$ elements conflicts with a down-neighbor $w \in \Lo(x)$ of $x$; therefore $c \sim \D_{N(x)}$ is a conflict with probability $\Pr[\kappa_{\chi^*, c, w}] \in O(\frac{1}{\Delta})$, so the expected amortized runtime $T_{\textsc{recolor}\textsubscript{$\log$}}$ of \textsc{recolor}\textsubscript{$\log$} is 
        \begin{align*}
            \E[T_{\textsc{recolor}\textsubscript{$\log$}}] = \frac{\Pr[\kappa_{\chi^*, c, w}]}{\Delta} \in O\left(\frac{\Delta}{\Delta}\right) = O(1).
        \end{align*}
        Then the total expected amortized update time for $\A_{\log}$ is in $O(\log \Delta)$.
        
    \subsection{Constant-time dynamic $\dpw$-coloring}
    \label{subsec:constant_dyn_d+1_bh}
        To improve the performance of dynamic $\dpw$-coloring from $O(\log \Delta)$ to $O(1)$ update time, \cite{bhattacharya2022fully} builds on \cite{bhattacharya2018dynamic} by disjoining \textsc{recolor}\textsubscript{constant} into two cases. One case corresponds to \textsc{recolor}\textsubscript{constant} from \cite{bhattacharya2018dynamic}; in the other case, when the recoloring palette is sufficiently small, a deterministic recoloring subroutine is used with time in $O(n)$.

        For this algorithm, $\A_{\text{constant}}$, the level data structure $\ell$ is spiritually the same as in $\A_{\log}$. It has levels $\{ -1, \ldots, \ceil{\log_3(n)} \}$ with associated data similar to those in $\A_{\log}$. Note that in both $\A_{\text{constant}}$ and $\A_{\log}$ there are numerous implementation details that are laid out in the papers. In this report I focus instead on the crucial algorithmic and formal details that distinguish these methods from others and from each other. The analysis and its apparatus in $\A_{\text{constant}}$ is significantly more involved than that of $\A_{\log}$, but the key invariant is similar.

        \begin{invariant}
            \label{inv:palette_size}
                If $v \in V$ was moved to level $\ell(v) \neq -1$ during a recoloring, the recoloring used a palette with size most $\frac{3^{\ell(v)} + 1}{2}$. If $\ell(v) = 1$ the palette had size 1.
            \end{invariant}

            $\A_{\text{constant}}$ has update time in $O(1 + n\frac{\log n + \Delta}{t})$ with high probability. For $t \in \Omega(n(\log n + \Delta))$ this is $O(1)$. The factor of $\Delta n$ comes from initializing the associated data structures that constitute the level data structure and in particular for each $v \in V$ there are $O(\Delta)$ elements of $\mathcal{D}_{N(v)}$. The factor of $n \log n$ comes from the total update time of vertices at levels containing at least a constant fraction of vertices that are recolored because they are at some point the more recently recolored endpoint of a conflicting inserted edge. The level data structure is reproduced in Figure \ref{fig:LDSconst} with several new definitions, including some artifacts of the analysis whose explanation I omit; the colored terms are defined in \ref{subsubsec:analysis_const}. In addition to numbering levels differently from $\A_{\log}$, proving constant update time with high probability for $\A_{\text{constant}}$ involves tracking which vertices are recolored at an \textit{original} root call to \textsc{recolor}\textsubscript{constant} --- endpoints of an inserted edge that conflicts with $\chi^*$ --- or from an \textit{induced}, descendant call to \textsc{recolor}\textsubscript{constant}. 

            \begin{figure}[t]
                \centering
                \begin{tikzpicture}[x=1cm,y=1cm,>=Latex, font=\small]
            
                    \colorlet{cInduced}{red!10}
                    \colorlet{cOriginal}{blue!10}
                    \colorlet{cFinal}{green!10}
            
                    \colorlet{nInduced}{red!50}
                    \colorlet{nOriginal}{blue!50}
                    \colorlet{nFinal}{green!50!lime}
            
                    \tikzstyle{inode}=[circle, draw=gray!40, inner sep=1.0pt, font=\tiny]
            
                    \def\X{-7}    
                    \def\Xr{1}    
                    \def\XC{-3}   
                    \def\yTop{3}
                    \def\yMid{0}
                    \def\yBot{-3}
                    \def\bandH{0.8} 
                    \def\legX{1.5} 
            
            
                    \fill[cInduced] (\X, \yTop-\bandH) rectangle (\Xr, \yTop+\bandH);
                    
                    \pgfmathsetseed{1}
                    \foreach \i in {1,...,30} {
                        \node[inode, fill=nInduced] at (\XC + rand*3.8, \yTop + rand*0.6) {};
                    }
                    \foreach \i in {1,...,5} { \node[inode, fill=nOriginal] at (\XC + rand*3.8, \yTop + rand*0.6) {}; }
            
                    \node[anchor=west, text width=5.5cm, align=justify, color=red!40!black, font=\footnotesize, inner sep=2pt] 
                        at (\legX, \yTop) {
                        \textbf{Induced-Heavy}\\[2pt]
                        $\geq \frac{1}{2}$ of the vertices, the \textcolor{nInduced}{red} ones, have \textit{induced} epochs. These are terminated because an up-neighbor was recolored. ($\ge 1/2$ epochs are induced).
                        };

                    \fill[cOriginal] (\X, \yMid-\bandH) rectangle (\Xr, \yMid+\bandH);
            
                    \pgfmathsetseed{2}
                    \foreach \i in {1,...,35} {
                        \node[inode, fill=nOriginal] at (\XC + rand*3.8, \yMid + rand*0.6) {};
                    }
                    \foreach \i in {1,...,5} { \node[inode, fill=nInduced] at (\XC + rand*3.8, \yMid + rand*0.6) {}; }
                    \foreach \i in {1,...,5} { \node[inode, fill=nFinal] at (\XC + rand*3.8, \yMid + rand*0.6) {}; }
            
                    \node[anchor=west, text width=5.5cm, align=justify, color=blue!40!black, font=\footnotesize, inner sep=2pt] 
                        at (\legX, \yMid) {
                        \textbf{Original-Heavy}\\[2pt]
                        $\geq \frac{3}{8}$ of the vertices, the \textcolor{nOriginal}{blue} ones, have \textit{original} epochs. These terminate due to the insertion of a same-color edge the more recently recolored endpoint of which is on this level.
                        };

                    \fill[cFinal] (\X, \yBot-\bandH) rectangle (\Xr, \yBot+\bandH);
            
                    \pgfmathsetseed{3}
                    \foreach \i in {1,...,20} {
                        \node[inode, fill=nFinal] at (\XC + rand*3.8, \yBot + rand*0.6) {};
                    }
                     \foreach \i in {1,...,15} {
                        \node[inode, fill=nOriginal] at (\XC + rand*3.8, \yBot + rand*0.6) {};
                    }
                     \foreach \i in {1,...,5} {
                        \node[inode, fill=nInduced] at (\XC + rand*3.8, \yBot + rand*0.6) {};
                    }
            
                    \node[anchor=west, text width=5.5cm, align=justify, color=green!30!black, font=\footnotesize, inner sep=2pt] 
                        at (\legX, \yBot) {
                        \textbf{Final-Heavy}\\[2pt]
                        Contains a significant portion ($\ge 1/8$) of \textcolor{nFinal}{green} vertices representing \textit{final} epochs, which are done changing color, and $<\frac{1}{2}$ of the vertices are induced.
                        };

            
                    \node[anchor=east] at (\X,\yTop+1.8) {$L = \lceil\log_3(n - 1)\rceil$};
                    \node at (\XC,\yTop+1.4) {\Large $\boldsymbol{\vdots}$};
            
                    \foreach \y/\lab in {
                        \yTop/{\ell(v)+1},
                        \yMid/{\ell(v)},
                        \yBot/{\ell(v)-1}
                    }{
                        \draw[gray!50] (\X,\y) -- (\Xr,\y); 
                        \node[anchor=east, fill=white, inner sep=1pt, text=black] at (\X,\y) {$\lab$};
                    }
            
                    \node at (\XC,\yBot-1.4) {\Large $\boldsymbol{\vdots}$};
                    \node[anchor=east] at (\X,\yBot-1.8) {$-1$};
            
                    \tikzstyle{snode}=[circle,draw,fill=white,inner sep=2pt, thick]
                    \node[snode] (v) at (\XC,\yMid) {$v$};
            
                    \begin{scope}[gray!70, thin, densely dashed]
                        \draw (v) -- (-4.5, \yTop - 0.3);
                        \draw (v) -- (-1.5, \yTop - 0.3);
                        
                        \draw[bend left=15] (v) to (-5.5, \yMid);
                        \draw[bend right=15] (v) to (-0.5, \yMid);
                        
                        \draw (v) -- (-4.5, \yBot + 0.3);
                        \draw (v) -- (-1.5, \yBot + 0.3);
                    \end{scope}
            
                \end{tikzpicture}
                \caption{Classification of levels based on epoch types. The small colored nodes represent individual epochs, colored by their termination cause. Vertex $v$ is shown at level $\ell(v)$, with weak dashed lines indicating its neighborhood in different levels.}
                \label{fig:LDSconst}
            \end{figure}
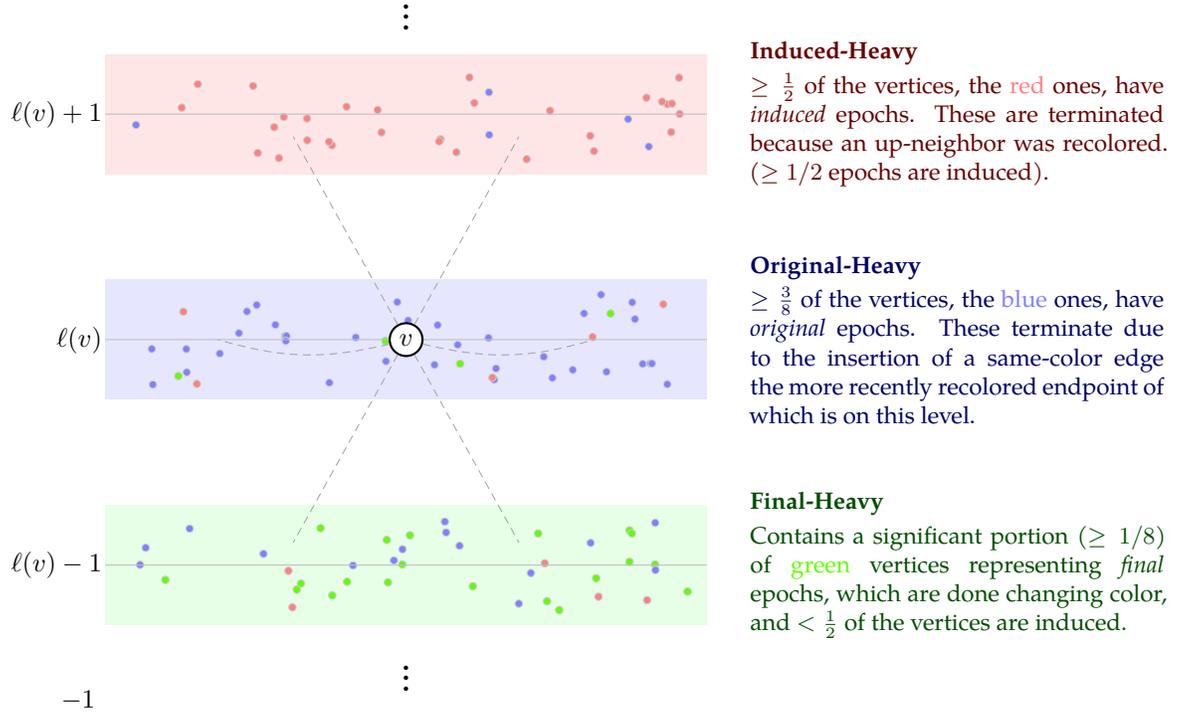

        \subsubsection{Algorithm}

        Let $\phi(v, \ell^*) = \sizeof{\{ u \in N(v) \st \ell(u) < \ell^* \}}$ be the number of neighbors of $v$ below level $\ell^*$ for each $v \in V$. \textsc{Edge insertion}\textsubscript{constant} in the nontrivial case of a conflicting edge is almost the same as \textsc{Edge insertion}\textsubscript{$\log$}, without the step to move vertices that break the invariants associated with $\A_{\log}$ between levels, and with retries when $\phi$ is too large. 

        \begin{algorithm}[H]
        \caption{\textsc{Edge insertion}\textsubscript{constant}}
        \label{alg:insert_const}
            \KwIn{$G$; $\{ u, v \} \leftrightsquigarrow E(G)$ for deletion or insertion; $\ell$ and associated data structures}
            \KwOut{Updated $\dpw$-coloring $\chi*$}
            \BlankLine
            \If{$\{ u, v \}$ is being inserted and $\chi^*(u) = \chi^*(v)$}{
                $x \gets \displaystyle \argmax_{x \in \{u, v\}} \tau_w$\tcp*{More-recently recolored}
                \textsc{recolor}\textsubscript{constant}($x$)\;
            }
        \end{algorithm}
            
        \begin{algorithm}[H]
        \caption{\textsc{recolor}\textsubscript{constant}}
        \label{alg:recolor_const}
            \KwIn{$G$; $v \in V(G)$ for recoloring; $\chi^*$; $\ell$ and associated data structures}
            \KwOut{$\chi^*$ with new color for $v$}
            \BlankLine
            \If{$\phi(x, \ell(x)) < 3^{\ell(x) + 2}$}{
                \textsc{Deterministic recolor}($x$)\;
                \Return null\;
            }
            \Else{
                \Return \textsc{Random recolor}($x$)\;
            }
        \end{algorithm}
        
        \begin{algorithm}[H]
        \caption{\textsc{Deterministic recolor}}
        \label{alg:det_recolor}
            \KwIn{$G$; $v \in V(G)$ for recoloring; $\chi^*$; $\ell$ and associated data structures}
            \KwOut{$\chi^*$ with new color for $v$}
            \BlankLine
            \For{$c \in \mathcal{D}_{N(v)}$}{
                \If{there is no vertex $w \in \Lo(v)$ with color $\chi^*(w) = c$}{
                    Set $\chi^*(v) \gets c$\;
                    Update down-neighbor data structure of $v$\;
                    Set $\ell(v) = -1$\;
                    \Return\;
                }
            }
        \end{algorithm}

        \begin{algorithm}[H]
        \caption{\textsc{Random recolor}}
        \label{alg:rand_recolor}
            \KwIn{$G$; $v \in V(G)$ for recoloring; $\chi^*$; $\ell$ and associated data structures}
            \KwOut{$\chi^*$ with new color for $v$}
            \BlankLine
            Set $\ell^\prime \gets \ell(v)$\;
            \While{$\phi(v, \ell^\prime + 1) \geq 3^{\ell^\prime + 2}$}{
                Increment $\ell^\prime$\;
            }
            Set $\ell(v) \gets \ell^\prime$\;
            Draw $c \sim \mathcal{D}_{N(v)}$ uniformly at random\;
            \If{$c \neq \chi^*(v)$}{
                Set $\chi^*(v) \gets c$\;
                Update down-neighbor data structure of $v$\;
            }
            \If{$c \in \mathcal{U}_{\Lo(v)} \setminus \mathcal{C}_{\Hi(v)}$}{
                Let $w \in \Lo(v)$ be such that $\chi^*(w) = c$\;
                \textsc{Recolor}\textsubscript{constant}($w$)\;
            }
        \end{algorithm}

        \subsubsection{Analysis}
        \label{subsubsec:analysis_const}
			The essence of the analysis for $\A_{\text{constant}}$ is the notion of an epoch $\Ep$, which characterizes the infrequency with which some vertices are recolored. Namely, epoch $\Ep$ corresponds to a vertex $x = v(\Ep)$ during which $x$ maintains its color $\chi^*(\Ep)$ and level $\ell(v)$. The cost $c(\Ep)$ of an epoch is the time required for the call to \textsc{Recolor}\textsubscript{constant}, applied to $v(\Ep)$, that began $\Ep$ by changing $\chi^*(v(\Ep))$ and $\ell(v(\Ep))$. When $\ell(x) = -1$ for $x = v(\Ep)$, and therefore $x$ is done being recolored, the analysis registers $c(\Ep)$. In other words, each vertex's contribution to $\A_{\text{constant}}$'s total update time comes from the sum of the cost of its epochs. An epoch is demarcated by either an original call to \textsc{Recolor}\textsubscript{constant}, when it is called original, or an induced call to \textsc{Recolor}\textsubscript{constant}, when it is called induced. See Figure \ref{fig:epochs}.

            \begin{figure}[t]
                \centering
                \begin{tikzpicture}[
                    thick,
                    >=Stealth,
                    font=\small,
                    event/.style={circle, fill=black, inner sep=2pt},
                    info/.style={align=center, font=\footnotesize}
                ]
            
                \draw[->] (0,0) -- (12,0) node[right] {Time $\tau$};
                \node[anchor=north east] at (0,0) {Node $v$};
            
                \coordinate (t1) at (1,0);
                \coordinate (t2) at (5,0);
                \coordinate (t3) at (9,0);
                \coordinate (tend) at (11.5,0);
            
                \node[event] (e1) at (t1) {};
                \node[event] (e2) at (t2) {};
                \node[event] (e3) at (t3) {};
            
				\node[below=0.3cm of e1, info] (c1) {\textsc{Recolor}\textsubscript{constant}($v$)\\Start $\Ep_1$};
                \node[below=0.3cm of e2, info] (c2) {\textsc{Recolor}\textsubscript{constant}($v$)\\Start $\Ep_2$\\Cost $c(\Ep_2)$};
				\node[below=0.3cm of e3, info] (c3) {\textsc{Recolor}\textsubscript{constant}($v$)\\Start $\Ep_3$\\Cost $c(\Ep_3)$};
            
                \draw[decoration={brace, amplitude=10pt, raise=5pt}, decorate]
                    (t1) -- (t2) node[midway, above=20pt, info] (epoch1) {
                        \textbf{Epoch $\Ep_1$} \\
                        $\ell(\Ep_1) = k \neq -1$ \\
                        $\chi(\Ep_1) = c_1$
                    };
            
                \draw[decoration={brace, amplitude=10pt, raise=5pt}, decorate]
                    (t2) -- (t3) node[midway, above=20pt, info] (epoch2) {
                        \textbf{Epoch $\Ep_2$} \\
                        $\ell(\Ep_2) = -1$ \\
                        $\chi(\Ep_2) = c_2$
                    };
            
                \draw[decoration={brace, amplitude=10pt, raise=5pt}, decorate]
                    (t3) -- (tend) node[midway, above=20pt, info] {
                        \textbf{Epoch $\Ep_3$} \\
                        $\dots$ \\
                        $\dots$
                    };
            
                \draw[->, dashed, red!70!black, line width=1pt]
                    (c2.south) .. controls +(down:2.5cm) and +(down:2.5cm) .. (epoch1.south)
                    node[midway, below=8pt, font=\scriptsize\itshape] {Cost charged to previous epoch};
            
                \end{tikzpicture}
                \caption{Epochs $\Ep_1, \Ep_2, \Ep_3$ for vertex $x = v(\Ep_1) = v(\Ep_2) = v(\Ep_3)$. $\Ep$ corresponds to the maximal time interval between consecutive recolorings of $x$. The cost of an epoch with level $\ell(\Ep) = -1$ ($\Ep_2$) is charged to the preceding epoch ($\Ep_1$). Here $\Ep_3$ could be induced.}
                \label{fig:epochs}
            \end{figure}
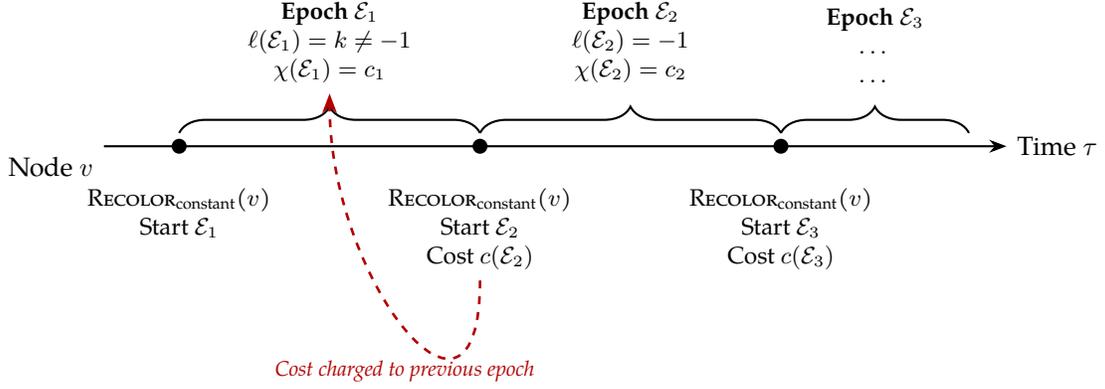

            Because no vertex $v(\Ep)$ changes level during $\Ep$, the set $\mathscr{E}$ of epochs can be partitioned across $\ell$. The bottleneck set of levels in $\ell$ is the set of \textit{original-heavy} levels, in which less than half of the vertices have induced epochs and less than an eighth have final epochs; that is, at least $\frac{3}{8}$ of the vertices in an original-heavy level have an original call to \textsc{Recolor}\textsubscript{constant}. The notation $c$ for cost is extended to the sum of costs of epochs in a level. 
            
            Then the main task of the analysis is to bound the cost of original-heavy levels, which is nontrivial because the cost of an original-heavy level depends on the fact that it is original-heavy. This means that Invariant \ref{inv:palette_size} does not help to lower bound the expected number of insertions before a conflicting edge is inserted analogously to how Invariants \ref{inv:down} and \ref{inv:up} ensured a sufficiently small probability of conflict for each edge insertion in $\A_{\log}$.

            To overcome this obstacle in the analysis, two time functions are defined. First, the duration $dur(\Ep)$ of epoch $\Ep$ is the number of edge insertions during $\Ep$. Second, the pseudo-duration $psdur(\Ep)$ of $\Ep$ is the number of distinct colors that neighbors of $v(\Ep)$ take during $\Ep$ until one of them takes the color $\ell(\Ep)$ that $v(\Ep)$ chose. For example, an adaptive adversary would ensure that the other endpoint, $u$, of each insertion $(v(\Ep), u)$ incident with $v(\Ep)$ has the same color $\chi^*(\Ep)$ as what $v(\Ep)$ took at the beginning of $\Ep$. See Figure \ref{fig:psdur}.

            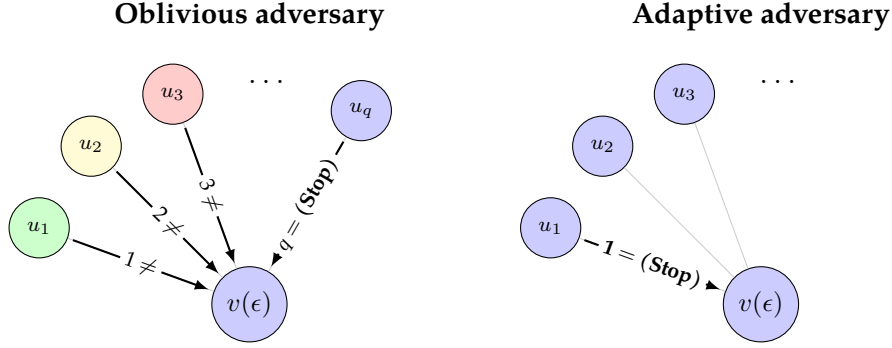
\begin{figure}[t]
                \centering
                \begin{tikzpicture}[
                    scale=0.85, 
                    >=latex, 
                    vertex/.style={circle, draw, minimum size=1cm, inner sep=0pt, font=\small},
                    neighbor/.style={circle, draw, minimum size=0.8cm, inner sep=0pt, font=\scriptsize},
                    faint edge/.style={-, black!20, thin},
                    notify arrow/.style={->, thick, shorten >=2pt, shorten <=2pt},
                    label box/.style={fill=white, inner sep=2pt, font=\scriptsize, midway, sloped},
                    dots/.style={font=\large}
                ]
            
                \begin{scope}
                    \node[font=\bfseries] at (0, 4.5) {Oblivious adversary};
            
                    \node[vertex, fill=blue!20] (v1) at (0, 0) {$v(\epsilon)$};
            
                    
                    \node[neighbor, fill=green!20] (n1) at (160:3.5) {$u_1$};
                    \draw[faint edge] (n1) -- (v1);
                    \draw[notify arrow] (n1) -- node[label box] {1 $\neq$} (v1);
            
                    \node[neighbor, fill=yellow!20] (n2) at (135:3.5) {$u_2$};
                    \draw[faint edge] (n2) -- (v1);
                    \draw[notify arrow] (n2) -- node[label box] {2 $\neq$} (v1);
            
                    \node[neighbor, fill=red!20] (n3) at (110:3.5) {$u_3$};
                    \draw[faint edge] (n3) -- (v1);
                    \draw[notify arrow] (n3) -- node[label box] {3 $\neq$} (v1);
            
                    \node[dots] at (85:3.5) {$\dots$};
            
                    \node[neighbor, fill=blue!20] (nq) at (60:3.5) {$u_q$};
                    \draw[faint edge] (nq) -- (v1);
                    \draw[notify arrow] (nq) -- node[label box, font=\scriptsize\bfseries] {$q$ $=$ (Stop)} (v1);
                \end{scope}
            
                \begin{scope}[xshift=8cm] 
                    \node[font=\bfseries] at (0, 4.5) {Adaptive adversary};
            
                    \node[vertex, fill=blue!20] (v2) at (0, 0) {$v(\epsilon)$};
            
                    
                    \node[neighbor, fill=blue!20] (a1) at (160:3.5) {$u_1$};
                    \draw[faint edge] (a1) -- (v2);
                    \draw[notify arrow] (a1) -- node[label box, font=\scriptsize\bfseries] {1 $=$ (Stop)} (v2);
            
                    \node[neighbor, fill=blue!20] (a2) at (135:3.5) {$u_2$};
                    \draw[faint edge] (a2) -- (v2);
            
                    \node[neighbor, fill=blue!20] (a3) at (110:3.5) {$u_3$};
                    \draw[faint edge] (a3) -- (v2);
            
                    \node[dots] at (85:3.5) {$\dots$};
                \end{scope}
            
                \end{tikzpicture}
                \caption{Example pseudo-duration, first in the oblivious case as in $\A_{\log}$, then how an adaptive adversary would exploit pseudo-duration to try to decrease epoch length.}
                \label{fig:psdur}
            \end{figure}
            
            By definition, $psdur(\Ep) \leq dur(\Ep)$, and it turns out that there is a low upper bound on the probability that $\Ep$ is short (\cite{bhattacharya2022fully} Lemma 3.11). Further, there is a low upper bound on the probability that a level has too many epochs while having more than a small number of short epochs. In fact, roughly speaking, the probability that any level has this undesirable trait is in $O(\frac{\log n}{n^a})$ for large constant $a$. This is responsible for the factor of $\log n$ in the total runtime, which is amortized out by sufficiently many edge insertions. Without this bound on the chance of epochs being too short in pseudo-duration, the deterministic runtime of $\A_{\log}$ would be in $O(tn^2 + \Delta n)$; in this case, even with $t \in \Omega(n(\log n + \Delta))$, $\A_{\log}$'s update time would have at least two extra factors of $n$.
            
    \section{Dynamic $\dpw$-coloring with adaptive adversaries} 
    \label{sec:adaptive}
        $\A_1$, $\A_2$, $\A_{\log}$, and $\A_{\text{constant}}$ assume oblivious inputs. That is, they assume inserted edges are drawn uniformly at random from the set $\binom{V(G)}{2} \setminus E(G)$ of non-edges of $G$. Obliviousness is congruous with various application paradigms in which inputs come from complicated real-world systems with no tractable distribution. \cite{behnezhad2025fully} examines $\dpw$-coloring that performs robustly against adaptively adversarial inputs, which break uniformity assumptions of the input and pessimally induce worst-case performance. An adaptive adversary $\Adv$ can stealthily inspect or simulate $\dpw$-colorer $\A$, determine $\A$'s performance on each permutation $\U \in \mathfrak{S}_{\mathscr{U}}$ of the desired inputs $\mathscr{U}$, and input the $\U$ that maximizes $\A$'s amortized update time to $\A$. To wit, \cite{behnezhad2025fully} presents a dynamic $\dpw$-colorer, $\A_\forall$, with update time in $\tilde{O}(n^{8/9})$ against such an adaptive adversary. To obtain sublinear time against $\forall$, $\A_\forall$ takes a phase-based approach that handles so-called sparse and dense vertices separately. Sparse vertices are recolored randomly at deterministic intervals so that their palettes typically have sufficiently many surplus colors; dense vertices are recolored repeatedly a bounded number of times in expectation until they are colored properly.
            
        In the previous $\dpw$-colorers, foreknowledge of $\Delta$ allows the algorithm to determine the number of buckets or height of the hierarchical vertex partition. In this case, using all $O(\Delta)$ available colors is an explicit invariant. 
        \begin{invariant}
        \label{inv:all_colors}
            Each color $c \in \{ 1, \ldots, \Delta + 1 \}$ is assigned to a number of vertices in $\tilde{O}(\frac{n}{\Delta})$.
        \end{invariant}
        With Invariant \ref{inv:all_colors}, verifying that a color is proper to assign to some vertex takes time in $\tilde{O}(\frac{n}{\Delta})$ instead of in $\Theta(\Delta)$. $\A_\forall$ uses a sparse--dense vertex decomposition to make it easier to track which colors are available for each vertex. The decomposition is roughly the Harris--Schneider--Su \cite{harris2016distributed} (HSS) decomposition, which partitions $V(G)$ into sparse vertices $V_S$ and dense vertices $V_D$. Let $G_S \coloneqq G[V_S]$ and $G_D \coloneqq G[V_D]$. $V_S$ is the set of vertices each $v$ of which has $\sizeof{E(G) \cap N(v)^2} \leq (1 - \epsilon^2)\binom{\Delta}{2}$; i.e., the neighborhood of each sparse vertex is $\epsilon^2$-far from being a $\Delta$-clique. It further partitions $G[V_D]$ into $C_1, \ldots, C_k$:
        \begin{align*}
            V(G) = V_S \oplus \bigoplus_{i = 1}^k V(C_i).
        \end{align*}
        Each $C_i \in \{ C_1, \ldots, C_k \}$ is an \textit{almost-clique}, so for some $\epsilon > 0$
        \begin{itemize}
            \item 
                $\sizeof{V(C_i)} \in [(1 - \epsilon) \Delta, (1 + \epsilon) \Delta]$ ($C_i$ has order close to $\Delta$)
        \end{itemize}
        and for any $v \in C_i$
        \begin{itemize}
            \item
                $\sizeof{V(C_i)} - \sizeof{N(v)} \leq  \epsilon \Delta$ ($v$ is adjacent to almost all of $C_i$) and
            \item
                $\sizeof{N(v) \cap V(G) \setminus V(C_i)} \leq \epsilon \Delta$ (few edges leave $C_i$).               
        \end{itemize}
        See Figure \ref{fig:hss}.
        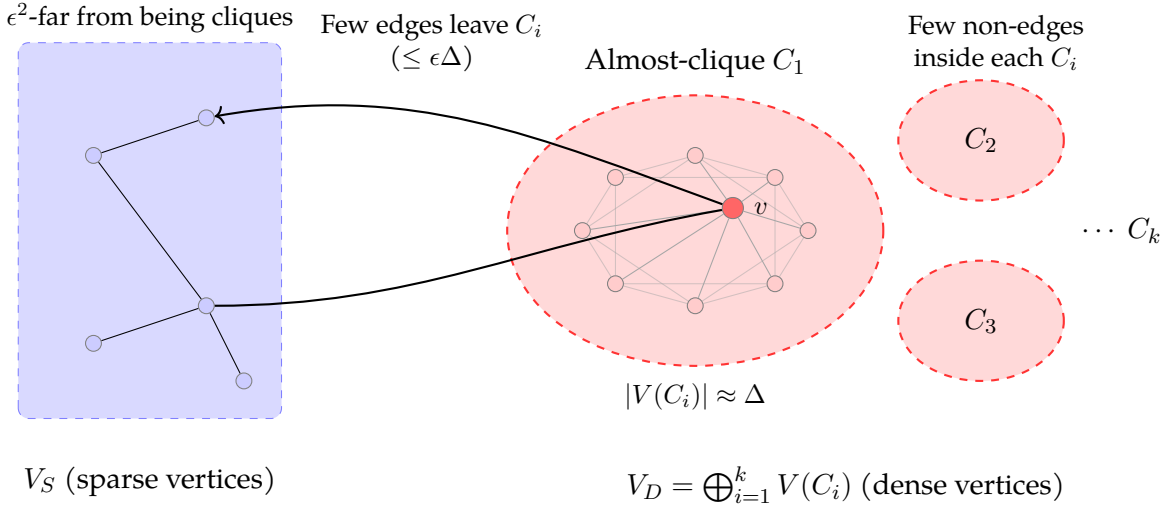
\begin{figure}[t]
            \centering
            \begin{tikzpicture}[
                vnode/.style={circle, draw=black!50, fill=gray!20, inner sep=0pt, minimum size=6pt},
                snode/.style={vnode, fill=blue!20}, 
                dnode/.style={vnode, fill=red!20},  
                clique region/.style={draw=red!80, thick, dashed, rounded corners, fill=red!15},
                vs region/.style={draw=blue!50, dashed, rounded corners, fill=blue!15},
                edge internal/.style={thin, gray!70},
                edge external/.style={thick, black},
                label text/.style={font=\small, align=center} 
            ]
        
            \draw[vs region] (-5, -2) rectangle (-1.5, 3);
            
            \node[anchor=north] at (-3.25, -2.5) {$V_S$ (sparse vertices)};
            
            \node[label text, anchor=south] at (-3.25, 3) {$\epsilon^2$-far from being cliques};
        
            \node[snode] (s1) at (-2.5, 2) {};
            \node[snode] (s2) at (-4, 1.5) {};
            \node[snode] (s3) at (-2.5, -0.5) {};
            \node[snode] (s4) at (-4, -1) {}; 
            \node[snode] (s5) at (-2, -1.5) {}; 
        
            \draw (s1) -- (s2);
            \draw (s2) -- (s3);
            \draw (s3) -- (s4);
            \draw (s3) -- (s5);
        
            \begin{scope}[xshift=4cm, yshift=0.5cm]
                
                \draw[clique region] (0,0) ellipse (2.5cm and 1.8cm);
                \node[anchor=south] at (0, 1.9) {Almost-clique $C_1$};
                \node[label text, anchor=north] at (0, -1.9) {$|V(C_i)| \approx \Delta$};
        
                \foreach \angle/\name in {0/c1a, 45/c1b, 90/c1c, 135/c1d, 180/c1e, 225/c1f, 270/c1g, 315/c1h} {
                    \node[dnode] (\name) at (\angle:1.5cm and 1cm) {};
                }
                \node[dnode, fill=red!60, minimum size=8pt, label={[font=\small]right:$v$}] (v) at (0.5, 0.3) {};
        
                \foreach \n in {c1a, c1b, c1c, c1e, c1f, c1g, c1h} \draw[edge internal] (v) -- (\n);
                \foreach \u/\w in {c1a/c1b, c1b/c1c, c1c/c1d, c1d/c1e, c1e/c1f, c1f/c1g, c1g/c1h, c1h/c1a,
                                    c1a/c1c, c1c/c1e, c1e/c1g, c1g/c1a,
                                    c1b/c1d, c1d/c1f, c1f/c1h, c1h/c1b} {
                    \draw[edge internal, opacity=0.5] (\u) -- (\w);
                }
        
                
                \begin{scope}[xshift=3.8cm, yshift=1.2cm]
                    \draw[clique region] (0,0) ellipse (1.1cm and 0.8cm);
                    \node at (0,0) {$C_2$};
                \end{scope}
        
                \begin{scope}[xshift=3.8cm, yshift=-1.2cm]
                    \draw[clique region] (0,0) ellipse (1.1cm and 0.8cm);
                    \node at (0,0) {$C_3$};
                \end{scope}
        
                \node[anchor=west] at (5cm, 0) {$\cdots ~C_k$};
        
            \end{scope}
        
            \node[align=center, anchor=north] at (6, -2.5) {$V_D = \bigoplus_{i=1}^k V(C_i)$ (dense vertices)};
        
        
            \node[label text] (note_external) at (8, 3) {Few non-edges\\inside each $C_i$};
        
            \node[label text] (note_external) at (0.5, 3) {Few edges leave $C_i$\\($\leq \epsilon \Delta$)};
            
            \draw[edge external, ->] (v) to[out=160, in=10] (s1);
            \draw[edge external] (v) to[out=190, in=0] (s3);
        
            \end{tikzpicture}
            \caption{HSS decomposition scheme. $V(G)$ is partitioned into sparse vertices $V_S$ and dense vertices $V_D$. $V_D$ is a disjoint union of almost-cliques $C_1, C_2, \dots, C_k$. Vertex $v \in C_1$ illustrates the property of being connected to almost all of its own clique with few edges leaving it.}
            \label{fig:hss}
        \end{figure}
        (There are a few more properties that hold of the almost-cliques involving adjustment complexity (see \cite{behnezhad2025fully} Theorem 2).)

        Decomposition is useful because $\A_\forall$ maintains colorings of $G_S$ and the $C_i$ independently. For $G_S$, $\A_\forall$ runs a one-shot refresh colorer that ensures each sparse vertex has $\Omega(\epsilon^2 \Delta)$ surplus colors in expectation. Namely, \textsc{One-shot sparse coloring} samples a random color for each vertex uniformly at random, then assigns the nonconflicting colors. See Algorithm \ref{alg:one_shot}. 
        Typically, this successfully recolors a constant fraction of $G_S$ and assigns identical colors to many neighbors of the typical vertex, thereby leaving $\Omega(\epsilon^2 \Delta)$ surplus colors for each vertex's palette. To ensure these colors are available, \textsc{One-shot sparse recoloring} is run after every $\Theta(\epsilon^2 \Delta)$ updates; each update recolors at most one vertex, so the surplus colors at vertex are reduced by at most one, giving remaining surplus colors in $\Omega(\epsilon^2 \Delta) - \Theta(\epsilon^2 \Delta) = \Omega(\epsilon^2 \Delta)$. In other words, for sparse vertices there are $\Theta(\epsilon^2 \Delta)$-update batches of updates, demarcated by \textsc{One-shot sparse recoloring}. Then to recolor a sparse vertex, the probability that one of these surplus vertices is selected uniformly at random is in $\frac{\Omega(\epsilon^2 \Delta)}{O(\Delta + 1)} = \Omega(\epsilon^2)$. From this, $O(\epsilon^{-2})$ such samples, each with a $\tilde{O}(\frac{n}{\Delta})$ feasibility check thanks to Invariant \ref{inv:all_colors}, will produce a good color, giving $\tilde{O}(\frac{n}{\epsilon^2 \Delta}) = n^{1 - \Omega(1)}$ --- sublinear --- recoloring time for sparse vertices.\footnote{After each batch begins with \textsc{One-shot sparse recoloring}, an extra greedy algorithm from \cite{assadi2025simple} colors the constant fraction of $V_S$ that was not successfully recolored.}

        \begin{algorithm}[H]
        \caption{\textsc{One-shot sparse coloring}}
        \label{alg:one_shot}
            \KwIn{Graph $G$; sparse vertices $V_S \subset V(G)$}
            \KwOut{Recolored vertices $V^\prime \subseteq V_S$}
            \BlankLine
            Initialize $V^\prime \gets \emptyset$\;
            \For{$v \in V_S$}{
                Sample $\chi^\prime(v) \sim [\Delta + 1]$ uniformly at random\;
            }
            \For{$v \in V_S$}{
                \If{$c(v) \neq \chi^*(w)$ for each $w \in N(v)$}{
                    Assign $\chi^*(v) \gets c(v)$\;
                    Add $v$ to $V^\prime$\;
                }
            }
            \Return{$V_S^\prime$}\;
        \end{algorithm}

        For $V_D$, dynamic coloring is more involved. The key technique is a dynamized static algorithm from \cite{assadi2019sublinear} in which each almost-clique $C_i$ is colored as follows. First, an approximate maximum matching of non-edges $\binom{C_i}{2} \setminus E(C_i)$ is found, and matching vertices are colored identically. In other words, as many nonadjacent vertices as possible are colored, using two vertices per color. Second, a perfect matching from the remaining vertices to the remaining colors is found. The matching is on a bipartite graph $\mathcal{H}$ with parts $\mathcal{V}$ (remaining vertices) and $\mathcal{C}$ (remaining colors), with $\{ v, c \} \in E(\mathcal{H})$ if $c \notin N(v) \setminus \mathcal{V}$. A perfect matching on $\mathcal{H}$ takes each vertex to an unused color, providing a proper coloring. The existence of the matching is nonobvious, and making this step dynamic without spending time linear in $\Delta$ (and therefore possibly linear in $n$, against $\forall$) is highly nontrivial.

        To dynamize this static algorithm, there are two cases. In the first case, for almost-cliques with order at least $\Delta + 1$, colors are partitioned into heavy and light colors. A color $c$ is heavy in almost-clique $C_i$ if there is a number in $\Omega(\Delta)$ of edges $(a_i, b_i)$ such that $a_i \in V(C_i), b_i \notin V(C_i), \chi^*(a_i) = c$. Light colors are not heavy. Heavy colors are less useful than light colors. Fortunately, $\mathcal{H}$ is still colorable if the heavy colors are excluded from $\mathcal{C}$, using an augmenting path that exists with high constant probability. In the second case, for almost-cliques with size less than $\Delta + 1$, there are two subcases. If the non-edge matching is sufficiently large (at least $\frac{\Delta}{10}$), $\A_\forall$ behaves similarly to the large almost-clique case, because enough of the vertices have already been handled. Otherwise, to recolor $v \in V(C_i)$, a random color $c$ yet unused in $C_i$ is picked uniformly at random. If $v$ has sparse neighbors, $c$ might cause a conflict, with probability in $\Omega(\frac{1}{k})$, where $k = \Delta + 1 - \sizeof{V(C_i)}$. Then in $\tilde{O}(\epsilon \Delta)$ such samples, $\A_\forall$ will typically find a usable color. If this does not work, then an augmenting path through $\mathcal{H}$ of length five, using three vertices $v_1, v_2, v_3 \in \mathcal{V}$ and three colors in $\mathcal{C}$, gives sufficiently large constant probability of properly coloring $v_1, v_2, v_3$. Finding the augmenting path takes the most runtime: $\tilde{O}(k)$ samples, each taking $\tilde{\frac{n}{\Delta}}$-time feasibility checks, contribute time in $\tilde{O}(\epsilon n)$ --- which is sublinear.

\section{Acknowledgments}
    Thanks to the authors for their contributions to the field and to Quanquan Liu for comments. Thanks also to OpenAI and Google's generative pre-trained transformer LLMs (Large \LaTeX~Mudslingers), which helped make the figures. 

\newpage
\printbibliography[title={References}]
\label{sec:references}

\end{document}